\documentclass[12pt,number]{elsarticle}
\usepackage{amsmath,amssymb,amsthm}%
\usepackage{microtype}
\usepackage{algorithm}
\usepackage{setspace}
\usepackage{color}
\usepackage{xcolor}
\usepackage[noend]{algpseudocode}
\usepackage{multirow}
\usepackage{paralist}
\usepackage{xspace}
\usepackage{tikz}
\usetikzlibrary{
  graphs,
  graphs.standard
}
\makeatletter
\let\originalleft\left
\let\originalright\right
\renewcommand{\left}{\mathopen{}\mathclose\bgroup\originalleft}
\renewcommand{\right}{\aftergroup\egroup\originalright}
\makeatother

\renewcommand{\Pr}{\operatorname*{\textbf{\textup{Pr}}}}

\DeclareMathOperator*{\EE}{\textbf{\textup{E}}}

\renewcommand{\geq}{\geqslant}
\renewcommand{\ge}{\geqslant}
\renewcommand{\leq}{\leqslant}
\renewcommand{\le}{\leqslant}
\newcommand{\set}[1]{\{ #1 \}}

\newcommand{\lt}{\left}
\newcommand{\rt}{\right}
\newcommand{\md}{\middle}

\newcommand{\congest}{\ensuremath{\mathsf{CONGEST}}}

\newboolean{short}
\setboolean{short}{false}
\newcommand{\onlyShort}[1]{\ifthenelse{\boolean{short}}{#1}{}}
\newcommand{\onlyLong}[1]{\ifthenelse{\boolean{short}}{}{#1}}

\newtheorem{lemma}{Lemma}
\newtheorem{theorem}{Theorem}
\newtheorem{claim}{Claim}
\newtheorem{corollary}{Corollary}

\newcommand{\ann}[1]{%
\text{\footnotesize(#1)}\quad}
\usepackage{bbding}

\newcommand{\msg}[1]{\langle #1 \rangle}
\newcommand{\local}{\ensuremath{\mathsf{LOCAL}}}
\newcommand{\kt}{\ensuremath{\mathsf{KT}}}
\newcommand{\fastIntervalMatching}{\ensuremath{\textsf{Fast-Interval-Matching}}}
\newcommand{\intervalConstruction}{\ensuremath{\textsf{Interval-Construction}}}
\newcommand{\uport}{\ensuremath{\mathsf{S}}}

\begin{document}
\begin{frontmatter}
\title{Perfect Matching with Few Link Activations}

\author[1]{Hugo Mirault}
\ead{hmirault@augusta.edu}

\author[1]{Peter Robinson}
\ead{perobinson@augusta.edu}

\author[1]{Ming Ming Tan}
\ead{mtan@augusta.edu}

\author[2]{Xianbin Zhu}
\ead{xianbin.aaronzhu@gmail.com}

\affiliation[1]{organization={School of Computer \& Cyber Sciences, Augusta University},
addressline={1120 15th Street},
city={Augusta},
postcode={30912},
state={GA},
country={USA}}
\affiliation[2]{organization={Department of Computer Science, Aalto University},
addressline={Otakaari 24},
city={Espoo},
postcode={02150},
country={Finland}}

\begin{abstract}
We consider the problem of computing a perfect matching problem in a synchronous distributed network, where the network topology corresponds to a complete bipartite graph.
The communication between nodes is restricted to activating communication links, which means that instead of sending messages containing a number of bits, each node can only send a pulse over some of its incident links in each round. 
In the port numbering model, where nodes are unaware of their neighbor's IDs, we give a randomized algorithm that terminates in $O\lt( \log n \rt)$ rounds and has a pulse complexity of $O\lt( n\log n \rt)$, which corresponds to the number of pulses sent over all links.
We also show that randomness is crucial in the port numbering model, as any deterministic algorithm must send at least $\Omega\lt( n^2 \rt)$ messages in the standard $\local$ model, where the messages can be of unbounded size.
Then, we turn our attention to the $\kt_1$ assumption, where each node starts out knowing its neighbors' IDs.
We show that this additional knowledge enables significantly improved bounds even for deterministic algorithms.
First, we give an $O\lt( \log n \rt)$ time deterministic algorithm that sends only $O\lt( n \rt)$ pulses.
Finally, we apply this algorithm recursively to obtain an exponential reduction in the time complexity to $O\lt( \log^*n\log\log n \rt)$, while slightly increasing the pulse complexity to $O\lt( n\log^*n \rt)$.
All our bounds also hold in the standard $\congest$ model with single-bit messages.

\end{abstract}

\begin{keyword}
distributed graph algorithm, message complexity, perfect matching
\end{keyword} 
\end{frontmatter}

\section{Introduction} \label{sec:intro}

We study the fundamental problem of computing a perfect matching among the nodes of a network, where the communication capabilities are limited to link activations, which is motivated by simple biological cell networks or sensor networks of resource-restricted devices.
Our model is inspired by protein-interaction networks and other biochemical signaling networks, where the alphabet of the messages between nodes are effectively restricted to a binary alphabet (i.e., signal or no signal), as demonstrated by \cite{cheong2011information}.

Concretely, we consider a complete bipartite graph consisting of two vertex sets $L$ and $R$, each consisting of $n$ nodes, and we assume that each node knows whether it is in $L$ or in $R$.
The communication between the nodes follows a \emph{link activation model}, which can be interpreted as a restricted variant of the standard \textsf{CONGEST} model~\cite{peleg}: 
In each synchronous round, each node may \emph{activate} any subset of its incident edges by sending a \emph{pulse} over these edges. 
That is, if a node $u$ sends a pulse over the edge $\set{u,v}$ in round $r$, then node $v$ observes this activation by $u$ at the start of round $r+1$.
We emphasize that these link activations are the only form of communication available in our model, which, in particular, prevents nodes from directly exchanging bit strings of information, as is commonly assumed in other models of distributed computation.
In particular, activating the link $\set{u,v}$ does \emph{not} allow $u$ to transmit a sequence of bits to $v$ in round $r$, but merely corresponds to a pulse that can be observed by $v$.

A model with similar features is the \emph{beeping model}, introduced in \cite{cornejo2010deploying}, where, in each synchronous round, a node may choose to ``beep'' across \emph{all} of its incident links, whereas, on the receiving side, each node can only distinguish between the case where none of its neighbors beeped or at least one of them did.
We refer the reader to \onlyLong{Navlakha and Bar-Joseph~}\cite{navlakha2014distributed} for a detailed survey describing the common aspects and differences of various biological systems and distributed computing models.

To quantify the performance of algorithms in the link activation model, we analyze the \emph{round complexity} (also known as \emph{time complexity}), which is the worst case number of rounds, as well as the \emph{pulse complexity}, whereby the latter corresponds to the worst-case number of pulses over all edges in any execution, for deterministic algorithms.
When considering randomized algorithms, we assume that each node has access to a private source of unbiased random bits that it may query during its local computation at the beginning of each round.
In that case, we also consider the \emph{expected pulse complexity}, where the expectation is computed over the random bits of the algorithm. 

While the pulse complexity is not a standard metric for distributed algorithms, it is closely related to the \emph{message complexity}. 
In fact, any upper bound on the pulse complexity also implies the same bound on the message complexity. 
Hence, designing algorithms with low pulse complexity immediately gives rise to message-efficient algorithms that only need to send a small number of single-bit messages in the $\congest$ model. 
Conversely, given an algorithm in $\congest$ model that sends only single-bit messages, one can simulate the algorithm in our link activation model with pulse complexity equal to the message complexity.

Our main goal of this work is to quantify the achievable round and pulse complexity for computing a perfect matching in the link activation model.
To this end, we distinguish two standard assumptions pertaining to the initial knowledge of the nodes:
\begin{compactenum} 
\item 
In the \emph{port numbering} or \emph{$\kt_0$ model}~\cite{peleg,DBLP:journals/csur/Suomela13}, we assume that nodes are anonymous and do not have unique IDs. Moreover, the incident edges to each node are numbered $1,\dots,n$ and, at the start of each round, every node $v$ only learns the subset of its incident port numbers that were activated by the respective other endpoint during the previous round.
A crucial difficulty in this setting is that the assignment of port numbers to communication links is not known in advance to the algorithm, and may be chosen adversarially. 
\item When considering the \emph{$\kt_1$ assumption}~\cite{AGPV88}, on the other hand, we assume that every node is equipped with a unique integer ID, chosen from some range of polynomial size, and an algorithm needs to satisfy the stated properties for all possible ID assignments to the nodes. Each node knows the IDs of all its neighbors, and, since we consider a complete bipartite graph, this means that every node in $L$ knows the IDs of all nodes in $R$, and vice versa. In particular, this allows a node in $L$ to directly activate a link to a node $v$ with some specific ID.
At the start of each round $r\ge 2$ every node $u$ learns the IDs of the nodes that activated a link to $u$ in round $r-1$.
\end{compactenum}

\onlyLong{\subsection{Our Contributions}}
\onlyShort{\medskip\noindent\textbf{Our Contributions.}}
We present several new algorithms and a lower bound for computing a perfect matching in complete bipartite graphs. As elaborated above, our results also hold in the $\congest$ model where the per-link bandwidth restriction is just one bit.\onlyShort{
The complete proofs of our results can be found in the full 
version of the paper~\cite{fullpaper}.
}
\begin{itemize} 
\item In Section~\ref{sec:up_bipartite}, we give a simple randomized algorithm that terminates in $O\lt( \log n \rt)$ rounds with a pulse complexity of $O\lt( n\log n \rt)$ with high probability. The algorithm assumes the port numbering model. 
\onlyLong{We also give an experimental evaluation of the performance of this algorithm in Section~\ref{sec:simulation} that shows that the asymptotic time complexity bound does not hide any large constants. }
\item \onlyLong{In Section~\ref{sec:kt1},}\onlyShort{Then,} we focus on the impact of the more powerful $\kt_1$ assumption, where nodes know the IDs of their neighbors, and show that this yields significantly improved bounds:
    \begin{itemize} 
    \item \onlyLong{In Section~\ref{sec:kt1_simple}, we}\onlyShort{We} show that the $\kt_1$ assumption is sufficient to circumvent the lower bound of Section~\ref{sec:lb}, by giving a deterministic algorithm that computes a perfect matching in $O\lt( \log n \rt)$ rounds and sends only $O\lt( n \rt)$ pulses.
    \item \onlyLong{In Section~\ref{sec:kt1_recursive}, we}\onlyShort{Next, we} leverage this deterministic algorithm\onlyLong{ from Section~\ref{sec:kt1_simple}} to obtain an exponential improvement in the number of rounds, yielding a time complexity of $O\lt( \log\log n \cdot \log^*n \rt)$, at the cost of a slightly increased pulse complexity of $O\lt( n \log^*n \rt)$.
    \end{itemize}
\item Finally, in Section~\ref{sec:lb}, we demonstrate that randomization is indeed crucial for achieving a low pulse complexity or message complexity in the port numbering model. We show that any deterministic algorithm must activate at least $\Omega\lt( n^2 \rt)$ links in the port numbering model. 
\onlyLong{
To strengthen the lower bound, we prove this result in the standard $\local$ model, where nodes can exchange messages of unbounded size in a single round.
This result stands in stark contrast to other symmetry-breaking problems such as leader election, which, as shown in \cite{DBLP:journals/dc/ChatterjeePR20}, can be solved deterministically in complete bipartite graphs in $O\lt( \log n \rt)$ rounds and $O\lt( n\log n \rt)$ messages, even in the $\congest$ model. 
}
\end{itemize}

\onlyLong{
\subsection{Other Related Work}

Closely related to the problem of computing a large matching is the research on distributed load balancing that has received significant attention during the past decades (please refer this book \cite{xu2007load} for more details). The basic problem is the well-known balls-into-bins problem: given $n$ balls and $n$ bins, how to place balls into bins as quickly as possible and the maximum bin load is small? The secret of success is $\textit{randomization}$. Under several assumptions, Lenzen and Wattenhofer \cite{DBLP:journals/dc/LenzenW16} developed tight bounds for parallel randomized load balancing improving the classical result in \cite{DBLP:journals/rsa/AdlerCMR98}, 
by showing that a maximum bin load of $2$ can be achieved in $O\lt( \log^*n \rt)$ rounds. 
While their setting is similar to ours in the sense that they consider a bipartite graphs, where ``balls'' are the nodes on the left and ``bins'' are the nodes on the right, it is unclear how to extend their result to obtain a perfect matching (i.e., a bin load of $1$) in a small number of rounds without sending too many messages.
In particular, one of the obstacles in solving the case of maximum bin load of $1$, is that it is no longer true that each ball has a constant probability (independent of others) to be sent to a bin that can accept the ball. This is, however, true for the case of maximum bin load of $2$.

In \cite{DBLP:conf/podc/Davies23a}, Davies  shows that a round of $\textsf{Broadcast CONGEST}$ can be simulated by noisy beeping model \cite{DBLP:journals/iandc/AshkenaziGL22} (and noiseless beeping model) within $O(\Delta \log n)$ rounds where $\Delta$ is the maximum degree in the network, which is near-optimal by the proposed lower bound. As an application of the above simulation, Davies gave a maximal matching algorithm in~$O(\Delta \log^2 n)$  rounds in the noisy beeping model and proved that maximal matching requires~$\Omega(\Delta \log n)$ rounds in the (noiseless) beeping model.

Similar to algorithms in the beeping model, content-oblivious algorithms are robust to noises. The motivation of content-oblivious algorithms is to design fault-tolerant algorithms that are robust to an unlimited amount of corruptions during communication. In \cite{DBLP:journals/dc/CensorHillelCGS23}, they give a method that can compile any distributed algorithms into content-oblivious ones over 2-edge-connected graphs (for other graphs, the noise can destroy any non-trivial computation). ~\cite{DBLP:conf/wdag/GhaffariK18a} removed the assumption that a distinguished node is required to start the computation. 

\section{Preliminaries}
\textbf{Notations. }
We follow the standard convention of the following notations.
We denote $[d]$ as the set of all integers from 
$1$ to $d$. 
We say that an event happens with high probability (w.h.p) if it occurs with probability at least $1-\frac{1}{n^{\Omega(1)}}$.
For $d \geq 0$, we denote $\log^{(d)}n$ as $d$-fold iterated logarithm of $n$, which is defined by applying the logarithm function $d$ times recursively to $n$, whereby $\log^{(0)}n=n$. 

In Section~\ref{sec:up_bipartite} we will make use of the following Chernoff bound for negatively correlated random variables: 
\begin{lemma}[\cite{impagliazzo2010constructive}] \label{lem:Chernoff_negative} 
Let $X_1, \ldots, X_d$ be (not necessarily independent) Boolean random variables and suppose that, for some $\delta \in [0, 1]$, it holds that, for every index set $S \subset [d]$, $Pr[\bigwedge_{i\in S} (X_i=1)] \leq \delta ^ {|S|}$. Then, for any $\eta \in [\delta, 1]$, we have $Pr[\sum_{i=1}^d X_i \geq \eta d]\leq e^{-2d(\eta-\delta)^2}$. 
\end{lemma}
}

\onlyLong{\section{An Algorithm for Perfect Matching in the Port Numbering Model}}
\onlyShort{\section{An Algorithm for Perfect Matching in $\kt_0$}}
\label{sec:up_bipartite}
Throughout this section, we consider a complete bipartite graph of $2n$ nodes, $G = (V, E)$ where the vertices are partitioned into two disjoint subsets $L$ and $R$. Nodes in $L$ are referred to as \textit{left} nodes, and nodes in $R$ are referred to as \textit{right} nodes. Note that each node knows whether it is a left or a right node. 
\onlyLong{
\begin{algorithm}
\small
\begin{algorithmic}[1]
    \State Set $T = \lceil\log_{\frac{1}{c}} n\rceil-\Theta(\log \log n)$ and $c = \frac{11}{12}$. 
    \State Initially, each node is \textit{unmatched}.
    \State \textbf{(Stage 1)}
     \For{phase $i = 1$ to $T$}
        \State In parallel, each unmatched left node selects $\left \lfloor \frac{1}{c^{i-1}} \right \rfloor$ random ports (with replacement) to send a $\msg{\texttt{prompt}}$ message. Each unmatched right node that receives a $\msg{\texttt{prompt}}$ message responds with an $\msg{\texttt{acknowledge}}$ message. 
        \State Each unmatched left node that receives $\msg{\texttt{acknowledge}}$ messages randomly selects one to respond with an $\msg{\texttt{invite}}$ message. 
        \State Each unmatched right node that receives $\msg{\texttt{invite}}$ messages randomly selects one to respond with a $\msg{\texttt{matched}}$ message.
        \State Each unmatched node that sends or receives a $\msg{\texttt{matched}}$ message is now matched. 
    \EndFor

    \State \textbf{(Stage 2)} 
    \State At phase $T+1$: 
     \State \hspace{\algorithmicindent} In parallel, each unmatched left node sends a $\msg{\texttt{prompt}}$ message to \emph{all} right nodes.
    \State \hspace{\algorithmicindent} Each unmatched right node that receives $\msg{\texttt{prompt}}$ messages responds with an $\msg{\texttt{acknowledge}}$ message.
    \State \hspace{\algorithmicindent} As a result, each unmatched left (resp. right) node $v$ knows the set $\uport_v$ of \emph{all} the ports that lead to the unmatched right (resp. left) nodes.
    
    \For{phase $i = T+2, \ldots,$ until all nodes have matched} 
        \State Each unmatched left node $v$ randomly selects one port in $\uport_v$ and sends an $\msg{\texttt{invite}}$ message. 
        \State Each unmatched right node that receives $\msg{\texttt{invite}}$ messages randomly selects one to respond with a $\msg{\texttt{matched}}$ message.
        \State As a result,  each unmatched node that sends or receives a $\msg{\texttt{matched}}$ message is matched. 
        \State \emph{Each} newly matched node $v$ sends a message to each port in $\uport_v$ to notify them that it has been matched. 
        \State Each unmatched node $v$ updates $\uport_v$ accordingly to exclude all ports that lead to newly matched nodes.
        \label{line:stage2notify}
    \EndFor
\end{algorithmic}
\caption{An algorithm for computing a perfect matching for complete bipartite graphs with $O(n \log n)$ pulse complexity and $O(\log n)$ time complexity. 
We point out that, in the pseudocode, we refer to messages of certain types, whereas in the link activation model, nodes can only send pulses. 
Since nodes have access to a synchronous global clock, they can interpret pulses received from specific nodes in a certain round as having a message of a specific type.}\label{alg:bipartite}
\end{algorithm}
}
\onlyLong{
In the following, we outline the high-level description of the randomized algorithm for Theorem \ref{alg:bipartite}, see Algorithm \ref{alg:bipartite}.
}
The algorithm is split into two stages. Stage $1$ consists of $T = \lceil\log_{\frac{1}{c}} n\rceil-\Theta(\log \log n)$ phases with the goal to reduce the number of unmatched nodes to $O(\log n)$. Stage $2$ deals with the remaining $O(\log n)$ unmatched nodes. Each phase consists of a constant number of rounds of communication. 
In phase $1$, each of the left nodes randomly selects one port to send a pulse to the right nodes. 
If a right node receives a pulse, then it will be matched by randomly selecting one left node from which it has received the pulse to match. 
We show that with high probability, after the end of phase $1$, the number of unmatched nodes reduces by $c:=\frac{11}{12}.$
Note that the node does not know which port leads to an unmatched neighbor until a pulse is exchanged between the two connected ports. 
Hence, a pulse sent by an unmatched left node might hit a right node that is already matched.
To ensure sufficient unmatched left nodes get matched, in phase $2$, each unmatched left node now sends $\lfloor 1/c \rfloor$ pulses. 
Again, we can show that with high probability, after the end of phase $2$, the number of unmatched nodes is reduced further by $c$. 
Ultimately, we can show that, for each subsequent phase $i\le T$, the unmatched nodes send $O(\frac{1}{c^{i-1}})$ pulses, and with high probability, the number of unmatched nodes is at most $c^in$ at the end of each phase. 
As a result, with high probability, the number of unmatched nodes after $T$ phases is $O(\log n)$. 

In Stage $2$, a matching for the remaining $O(\log n)$ unmatched nodes is constructed in $O(\log n)$ phases and using $O(n \log n)$ pulses. 
The key idea to achieve the desired number of pulses and the time complexity bound in Stage $2$ is that, by exchanging $O(n \log n)$ pulses in total, each unmatched node learns \emph{all} the ports that lead to its unmatched neighbors in each phase.

\begin{theorem} \label{thm:up_bipartite}
   There exists a randomized algorithm that, with high probability, constructs a perfect matching on a complete bipartite graph with $2n$ nodes in $O(\log n)$ time and sends $O(n \log n)$ pulses. 
  
\end{theorem}
\onlyLong{
In the following, we show that Algorithm \ref{alg:bipartite} achieves the stated time and pulse complexity with high probability.
We organize the proof as a series of lemmas.
Let $U_i$ be the number of unmatched right nodes (which is equal to the number of unmatched left nodes) at the end of phase $i$. Initially all  nodes are unmatched, hence, $U_0 = n$.

In our proof, we make use of a balls-into-bins argument where an $\msg{\texttt{invite}}$  message represents a ball, and an unmatched right node represents a bin. 
Note that, strictly speaking, the nodes cannot send ``messages'' carrying additional information, such as indicating that it is an $\msg{\texttt{invite}}$ message. 
However, we assume a synchronous clock, and thus whenever a node sees a pulse in some specific round, it can deduce the intended message type locally.

Let $B_k$ be the indicator variable that bin $k$ does not receive a ball. 

\begin{lemma} \label{lem:constant_accept}
Let $R$ be a set of integers such that for each bin $k$ in $R$, we have $\Pr[B_k=1]<\delta$ for some constant $\delta$. Then, it holds that for every subset $S$ of $R$, 
\begin{align*}
    \Pr\Big[ \bigwedge_{k\in S} (B_k = 1)\Big] \le \delta^{|S|}.
\end{align*}
\end{lemma}

\begin{proof}
    We observe that the probability that bin $k$ receives a ball conditioned on the event that bin $j$ does not receive a ball is at least the probability that bin $k$ receives a ball without conditioning. This means $\Pr[B_k=0|B_j=1] \geq \Pr[B_k=0]$ or equivalently,  $\Pr[B_k=1|B_j=1] \leq \Pr[B_k=1].$
    Hence, 
  \begin{align*}
    \Pr\Big[ \bigwedge_{k\in S} (B_k = 1)\Big]
    =
        \prod_{k \in S} \Pr\Big[B_k=1|\bigwedge_{j<k, j \in S} (B_j=1)\Big]
    \le
        \prod_{k \in S} \Pr[B_k=1]
    \le
        \delta^{|S|}.
  \end{align*}
\end{proof}

\begin{lemma}\label{lem:phase1}
With high probability, 
$U_1 < c\cdot n$. 
\end{lemma}

\begin{proof}
Note that in each phase, an unmatched right node will be matched at the end of the phase if it receives an $\msg{\texttt{invite}}$ message from a left node.
Hence, to bound the number of unmatched right nodes at the end of phase $1$, we bound the probability a given right node does not receive an $\msg{\texttt{invite}}$ message.  
We use the balls-into-bins analogy and observe there are $n$ balls that are being placed into $n$ bins uniformly at random. 
The probability that a bin never receives a ball is $(1-\frac{1}{n})^n \leq e^{-1}.$
Lemma \ref{lem:constant_accept} shows that 
\begin{align*}
    \Pr\Big[ \bigwedge_{k=1}^n (B_k = 1)\Big] \le e^{-|S|}.
\end{align*}
Now, we can use Lemma \ref{lem:Chernoff_negative} with $d = n, \delta = e^{-1}$ and $\eta = \frac{1}{2}$ to derive
$\Pr\Big[U_1 \geq \frac{n}{2}\Big] \leq e^{-2n(1/2-e^{-1})^2}$, which is at most $1/n^2$ for sufficiently large $n$. 
Hence, with high probability, $U_1$ is less than $\frac{n}{2}$, which is less than $c\cdot n$. 
\end{proof}

\begin{lemma}\label{lem:phasemiddle}
With high probability, $U_i < c^i n$ for every $1<i \leq T$.
\end{lemma}

\begin{proof}
We will show by induction that $U_i < c^i n$ (w.h.p.), assuming that the statement holds for $i-1$.
This suffices, because we can take a union over the $T=O(\log n)$ phases.

Note that the inductive basis ($i=1$) follows from Lemma~\ref{lem:phase1}. 
Now, condition on the event that 
\begin{align} \label{eq:condition_up}
    U_{i-1}< c^{i-1}n, 
\end{align}
we show that $U_i < c^i n$.
Note that since $U_i \leq U_{i-1}$, if $U_{i-1}<c^i n$, then we are done. 
Therefore, for the rest of the proof, we also condition on the event that 
\begin{align}\label{eq:condition_low}
U_{i-1} \geq c^i n.
\end{align}

We say that a left node $v$ \emph{discovers} a right node $u$ if $v$ receives an $\msg{\texttt{acknowledge}}$ message from $u$. 
Let $S_i$ be the set of unmatched left nodes that discovers \emph{exactly} one unmatched right node.  
Let $A_i$ be the size of $S_i$. 
We show that, with high probability, 
\begin{align}\label{eq: A_i}
A_i \geq \frac{3}{4} \cdot c \cdot e^{-2}  \cdot U_{i-1}.
\end{align}
Consider a left node $v$. 
\begin{align*}
    \Pr[v \in S_i]
     & = \lfloor1/c^{i-1}\rfloor 
     \left(\frac{U_{i-1}}{n}\right)
     \left(1-\frac{U_{i-1}}{n}\right)^{\lfloor1/c^{i-1}\rfloor -1} 
     \\
     \ann{by \eqref{eq:condition_low} and \eqref{eq:condition_up}}
     & > \lfloor1/c^{i-1}\rfloor
     c^i 
     \left(1-c^{i-1}\right)^{\lfloor1/c^{i-1}\rfloor-1} \\
     & \geq (1/c^{i-1}-1)  c^i
      \left(1-c^{i-1}\right)^{\lfloor(1/c)^{i-1}\rfloor-1}\\
    & = c\left(1-c^{i-1}\right)^{\lfloor(1/c)^{i-1}\rfloor} \\
    & > c \cdot e^{-2}. 
  \end{align*}
Hence, $\EE\lt[A_i\rt] \geq c\cdot e^{-2} \cdot U_{i-1}.$ 
It follows from \eqref{eq:condition_low} and  $i \leq sT = \lceil\log_{\frac{1}{c}} n\rceil-\Theta(\log \log n)$ that $\EE\lt[A_i\rt]$ is $\Omega(\log n)$. 
Note that $A_i$ can be expressed as the sum of some  independent indicator random variables, since the left nodes select the ports to send the $\msg{\texttt{prompt}}$ messages independently. 
Consequently, we can apply a standard Chernoff's Bound to derive \eqref{eq: A_i}. 

For the remainder of the proof, we condition on the event that \eqref{eq: A_i} holds.
We observe that an unmatched right node $u$ that is discovered by a node $v$ in $S_i$ must receive an $\msg{\texttt{invite}}$ message from $v$, and hence $u$ will be matched (not necessarily to $v$) at the end of the phase. 
Hence, for a right node $u$ to remain unmatched at the of the phase, it must be the case that $u$ is not discovered by any node in $S_i$. 
This implies that the probability that a right node $u$ is not matched is at most the probability that it is not discovered by any node in $S_i$, which is 
$\left (1-\frac{1}{U_{i-1}}\right)^{A_i} \leq e^{-\frac{A_i}{U_{i-1}}}$.
It follows from \eqref{eq: A_i} that this is at most $e^{-\frac{3}{4}\cdot c \cdot e^{-2}}.$

Similarly as in the proof of Lemma~\ref{lem:phase1}, we use Lemma~\ref{lem:constant_accept} and Lemma~\ref{lem:Chernoff_negative} with $d = U_{i-1}, \delta = e^{-\frac{3}{4}\cdot c \cdot e^{-2}}$ and $  \eta = c$ to derive that
  \begin{align*}
    \Pr\Big[U_{i} \geq c \cdot U_{i-1}\Big] 
    \leq
        e^{-2U_{i-1}(c-\delta)^2}\\ 
    \ann{by \eqref{eq:condition_low}}
    \leq 
        e^{-2\cdot c^i n (c-\delta)^2}.
  \end{align*} 
When $i \le T$, we have $\Pr\Big[U_{i} \geq c \cdot U_{i-1}\Big]<\frac{1}{n^2}$ for sufficiently large $n$.
Hence, with high probability, $U_i < c \cdot U_{i-1}$. It follows from \eqref{eq:condition_up} that $U_i < c^i n$. 
 \end{proof}

\begin{lemma}\label{claim:phaselast}
After phase $T$, a perfect matching for the remaining unmatched nodes will be constructed in $O(\log n)$ phases with high probability.
\end{lemma}

\begin{proof}
From Lemma~\ref{lem:phasemiddle}, we have $U_T=O(\log n)$ with high probability.
For the remainder of the proof, we condition on the event $U_T=O(\log n)$.
Let $U^*$ be the set of unmatched left nodes at the end of phase $T$. 
At each phase of Stage $2$ of the algorithm, every unmatched right node knows the ports that lead to $U^*$ and vice versa. 
Let $\EE\lt[U_i \md| U_{i-1}\rt]$ be the expected number of unmatched nodes at the end of phase $i$ given that the number of unmatched nodes at the beginning of phase $i$ is $U_{i-1}$.

     Let $i>T$.  
     Since each unmatched left node knows the port that connects to the unmatched right nodes, each unmatched left node chooses $1$ out of $U_{i-1}$ ports randomly to send a message.
     Hence, there are $U_{i-1}$ balls that are placed uniformly random into $U_{i-1}$ bins. 
     Therefore, the probability that a given right bin does not get a ball is at most $\Big(1-\frac{1}{U_{i-1}}\Big)^{U_{i-1}} \leq e^{-1}.$
     By linearity of expectations, we have $E[U_i \ | \ U_{i-1}]\leq e^{-1}\cdot U_{i-1}$, and thus 
    \begin{align*}
        \EE\lt[U_i\rt] 
        = 
            \EE\lt[\EE\lt[U_i \ \md| \ U_{i-1}\rt]\rt] 
        \leq 
           e^{-1}\cdot\EE\lt[U_{i-1}\rt].
    \end{align*} 
   Hence, for $j\geq 1$, 
   \begin{align*}
        \EE\lt[U_{T+j}\rt] 
        \leq  
            \EE\lt[U_T\rt](e^{-1})^j 
        \leq 
            U_T(e^{-1})^j. 
    \end{align*} 
By assumption, we have $U_T  \leq c\log n$ for some constant $c>0$. 
Let $j = 4 \log n$ and $f=T+j$.  we have 
\begin{align*}
    \EE\lt[U_{f}\rt] 
    \leq  
        \frac{c\log n}{n^4}
    \leq 
        \frac{c}{n^2}
\end{align*} 
Using Markov's inequality, we can bound the number of unmatched nodes after phase $f$:
\begin{align*}
   \Pr(U_f \geq 1)
   \leq 
    \EE\lt[U_{f}\rt]
   \leq 
    \frac{c}{n^2}.
\end{align*} 
Hence, with high probability, after $f = O(\log n)$ phases, the algorithm terminates and a perfect matching is obtained.
\end{proof}

\begin{lemma}
    The number of pulses of Algorithm \ref{alg:bipartite} is $O(n \log n)$.
\end{lemma}

\begin{proof}
Let $i\leq T$. Each unmatched node at phase $i$ sent $O(\frac{1}{c^{i-1}})$ pulses. By Lemma~\ref{lem:phase1} and Lemma~\ref{lem:phasemiddle}, the number of unmatched nodes at phase $i$ is at most $c^in$. As a result, $O(n)$ pulses are sent in phase $i$. 
This gives $O(n\log n)$ pulse complexity for Stage $1$ of the algorithm. 
For $i=T+1$, each unmatched left nodes sends $O(n)$ pulses,  which amounts to $O(n\log n)$ pulses, since the number of unmatched left nodes is $O(\log n)$.  
Let $U^*$ be the set of unmatched left nodes at the end of phase $T$. 
For $i \geq T$, at the end of each phase, each newly matched right node sends pulses to notify the nodes in $U^*$ that it has been matched. This generates $O(\log^2 n)$ messages since the number of unmatched nodes at each phase is bounded by $|U^*| = O(\log n)$. 
Hence, the pulse complexity for Stage $2$ of the algorithm is $O(n \log n)$ as well. 
\end{proof}
}

\onlyLong{\subsection{Experimental Results} \label{sec:simulation}
In Theorem~\ref{thm:up_bipartite}, we have that the running time of the randomized algorithm is logarithmic in the number of nodes. 
That is, there exists a constant $d$ such that Algorithm \ref{alg:bipartite} terminates in $d\log n$ rounds. 
In the following, we show that the constant $d$ is small. 
We implemented Algorithm \ref{alg:bipartite} and tested the performance of the algorithm on complete bipartite graphs of $n$ nodes, for $n = 2^i$, where $i=1, \ldots, 20$. 
For each case, we performed 1000 trials, and computed the mean number of phases taken by the algorithm to terminate. Each phase consists of at most $4$ rounds of communications. 
Figure \ref{fig:simulation} presents the experimental results, which indicate that the algorithm terminates in $\log_2 n$ phases where $n$ is the number of nodes. 

\begin{figure}[th]
  \centering
    \includegraphics[scale=0.3]{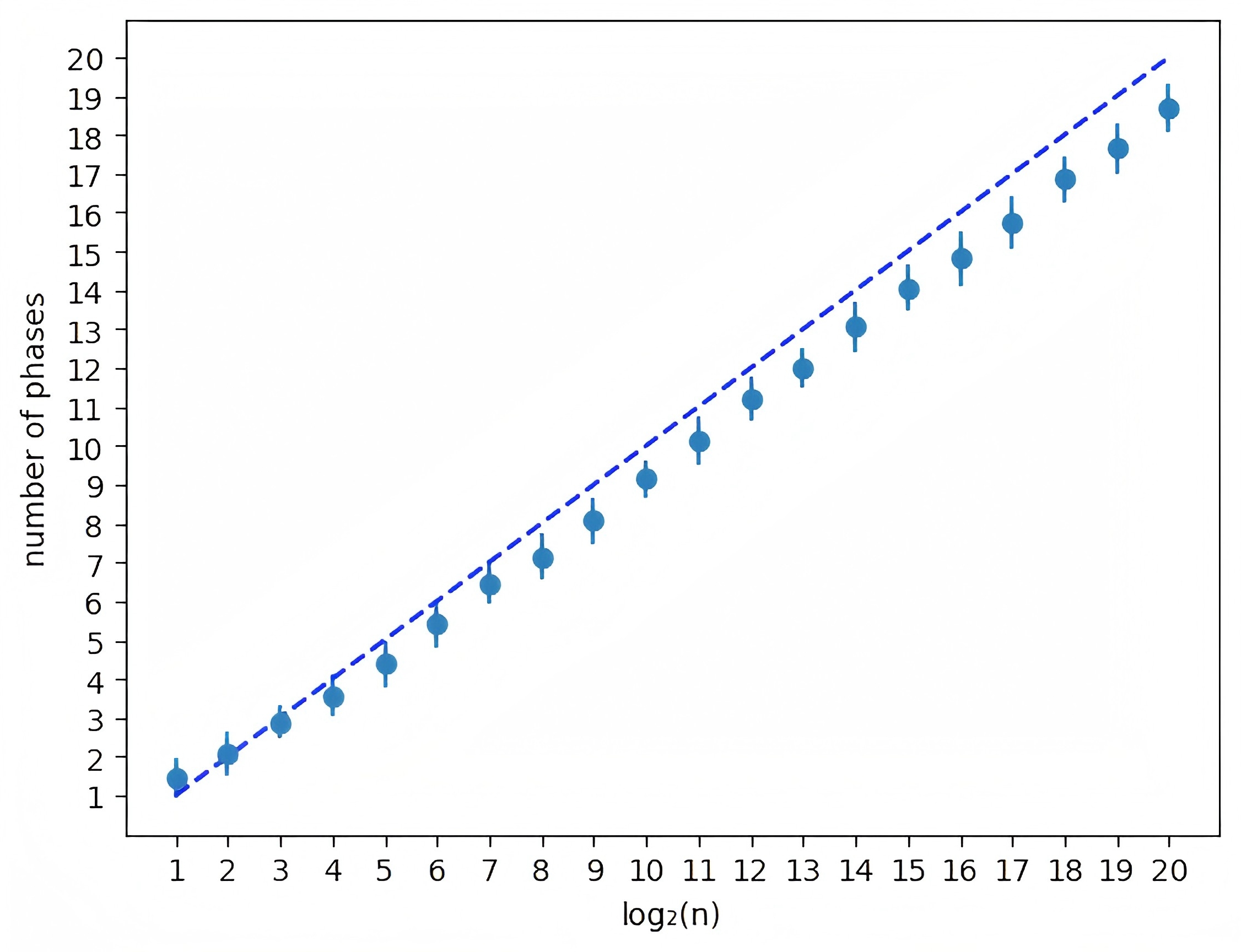} 
  \caption{Experimental results of the execution of Algorithm \ref{alg:bipartite}. The x-axis represents the logarithmic base $2$ of the number of nodes, and the y-axis represents the number of phases. Each phase consists of at most $4$ communication rounds. The points represent the mean number of phases of 1000 trials in each case, and the error bars indicate the standard deviation. The dashed line is the identity line $y=x$.
  }\label{fig:simulation}
\end{figure}

 }
\section{More Efficient Algorithms under the $\kt_1$ Assumption} \label{sec:kt1}
In this section, we present two deterministic algorithms that leverage the knowledge of their neighbors' IDs to not only break the lower bound of Section~\ref{sec:lb}, but also improve over the performance of the randomized algorithm in the port numbering model given in Section~\ref{sec:up_bipartite}.

Here, we assume the standard $\kt_1$ model, where nodes have unique integer IDs, and each node starts out knowing the IDs of its neighbors. 
We say that node $u \in L$ (resp. $R$) has \emph{rank $i$}, if its ID is the $i$-th smallest among all nodes in $L$ (resp. $R$). 
Note that due to $\kt_1$ assumption, each node in $L$ can locally compute the rank of each node in $R$ and vice versa. 
However, each node is unaware of its own rank. 
We remark that if each node knew its own rank (which is the case if we assume $\kt_2$~\cite{AGPV88}, where each node knows the IDs of all nodes at most two hops away from it), then all nodes can output a matching instantly without any communication. 
Returning to our $\kt_1$ setting, this motivates the algorithmic design strategy of informing each node of its own rank. 
In the $\congest$ model, where each message can carry $O(\log n)$ bit, it takes just two rounds and $O(n)$ messages to inform each left node of its rank. Concretely, in the first round, each left node sends a message to the smallest rank right node $r_0$. In the second round, $r_0$ sends one message to each of the left nodes to inform each of them of their rank. 
On the other hand, in our link activation model, where we can only signal a pulse, a naive implementation of the above approach would take $O(\log n)$ rounds and $O(n \log n)$ pulses, which would not improve over the port numbering algorithm of Section~\ref{sec:up_bipartite}: 
When using the convention that a $0$ bit sent over a link corresponds to non-activation of the link, and a $1$ bit corresponds to a pulse,
in round $i$, node $r_0$ sends the $i$-th bit of the binary representation of $j$ to the left node of rank $j$. This can be completed in $O(\log n)$ rounds. %
\footnote{An alternate approach for each left node to learn its own rank is by performing the standard binary search. That is, each left node $l$ initially guesses its rank $i$ and sends a pulse to the right node $r$ of rank $i$. In two rounds, $r$ can inform $l$ if its guess is equal, higher, or lower than the correct value. \onlyLong{The standard binary search analysis implies that each left node takes $O(\log n)$ rounds to find its rank, and hence, $O(n \log n)$ pulses.}}
\onlyLong{
\subsection{Algorithm $\fastIntervalMatching$} \label{sec:kt1_simple}
}
We now describe the algorithm $\fastIntervalMatching$, which achieves optimal pulse complexity of $O(n)$ and terminates in $O(\log n)$ rounds. 

\medskip\noindent\textbf{Interval Construction:} We describe a procedure called $\intervalConstruction$ that partitions the nodes in $L$ and $R$ into equal size intervals. As we will instantiate this method for subgraphs of different sizes, we state it under the assumption that $|L|=|R|=N$, where $N$ is any positive integer that is greater than $2$.
We use $l_0,\dots,l_{N-1}$ to denote the nodes in $L$ ordered in increasing order of their IDs, and analogously define the order $r_0,\dots,r_{N-1}$ of the nodes in $R$. 
Note that this means the rank of $l_i$ as well as $r_i$ is $i$. 
Each node in $L$ activates a link to the node $r_0$, who has the smallest ID in $R$. 
We define $s = \lt\lceil \log N \rt\rceil$.

\emph{Helper Inform Step:} $r_0$ recruits a set $H$ of $s$ \emph{helpers}, by contacting the nodes with the $s$-smallest ranks in $L$, and informing these nodes of their respective rank in $L$. 
In more detail, $r_0$ executes $O(\log s)$ rounds, where in each round $i$, $r_0$ sends the $i$-th bit of rank $j$ to $l_j$, for each $j=0, \ldots, s-1$ in parallel. 

\emph{Partition Step:}
Our next goal will be to split $R$ into $N/s$ intervals,\footnote{To simplify the notation, we assume that $N/s$ is an integer.} denoted by $R_0,\dots,R_{N/s-1}$.
For every $i \in [0,N/s-2]$, interval $R_i$ contains exactly $s$ nodes, whereby 
$R_{N/s-1}$ contains at most $s-1$ nodes.
We will obtain the intervals $L_0,\dots,L_{N/s-1}$ by partitioning $L$ analogously.
To implement this partitioning process, we make use of the helper nodes to identify every node $r_{j\cdot s }$, for $j \in \set{0,\dots,n/s-1}$, as the \emph{leader of interval $R_j$}.
To this end, we need to make sure that each of these interval leaders learns its own rank.
Concretely, for every $r_{j\cdot s}$, each helper $l_i$ sends the $i$-th bit of $r_{j\cdot s}$' rank to $r_{j\cdot s}$ (where $0$ and $1$ bits are encoded as pulses or absence of pulses, as described above).
Note that all helpers perform this operation in parallel, and thus this requires just $1$ round.
Afterwards, each interval leader $r_{j\cdot s}$ activates a link to $l_{j\cdot s}$.
Since the \emph{leader of $L_j$}, which is $l_{j\cdot s}$, knows the rank of $r_{j\cdot s}$, a pulse is sufficient for it to learn its own rank. 
Then, we add all edges $\set{l_{j\cdot s},r_{j\cdot s}}$ to the matching, for each $j$. 

\emph{Interval Inform Step:} In one additional round, the leader of $L_j$ sends a pulse to each node in interval  $R_j$, and similarly the leader of $R_j$ sends a pulse to the nodes in the $L_j$, for every $j$.
This completes the $\intervalConstruction$ procedure.

\onlyLong{\medskip\noindent\textbf{Remark:} Strictly speaking, algorithm $\fastIntervalMatching$ does not require the Interval Inform Step, as we will see in the analysis below. We nevertheless include this step, as it does not change the asymptotic complexity bounds, and our second algorithm (see Sec.~\ref{sec:kt1_recursive}) is based on a recursive application of $\intervalConstruction$ procedure, in which case, the Interval Inform Step turns out to be necessary.  
}

\medskip\noindent\textbf{Parallel Interval Matching:}
Finally, we compute a matching in parallel on each interval: 
Each leader in $R_j$, i.e., node $r_{j.s}$, sends a pulse to the nodes in $L_j$ sequentially, by activating one link per round. 
This allows each node in $L_j$ to deduce its rank relative to the other nodes in $L_j$.
In particular, node $l_{j\cdot s+ k}$ will have its link activated after $k$ rounds, prompting it to activate the link to $r_{j\cdot s + k}$, to which it becomes matched.
Since there are $s$ nodes per interval this completes in $O\lt( s \rt)$ rounds.
\onlyLong{
\begin{lemma} \label{lem:kt1_messages}
The $\intervalConstruction$ procedure has $O(N)$ pulse complexity and completes in $O(\log \log N)$ rounds.
\end{lemma}
\begin{proof}
It is straightforward to verify that the round complexity of this procedure is $O\lt( \log s \rt) = O(\log \log N)$. 
Hence, we focus on the pulse complexity analysis. 
Recall that each interval contains $s:=\lceil \log N \rceil$ nodes and there are $N/s$ left (and also the same number of right) intervals. 
In the \textsl{Helper Inform Step}, the leader $r_0$ informs each of the helper nodes of their respective ranks. Each rank can be encoded as a $\log_2 s$ bit string since the highest rank of these nodes is $s$.   
Hence, $r_0$ sends at most $s$ pulses per round, for a consecutive of $\log_2 s$ rounds, which add up to $O(\log N \cdot \log \log N)$ pulses in total. 
In the \textsl{Partition Step}, the right interval leaders are informed of their ranks. Hence, each right leader receives at most $O(s)$ pulses. Since there are $N/s$ right leaders, the pulse complexity is $O(N)$. Note that each left leader learns its own rank by receiving just one additional pulse from their respective right leaders of the same rank. (The nodes on the left who do not receive a pulse in that round know that they are not the leader of any interval.) Hence, the overall pulse complexity of the \textsl{Partition Step} is $O(N)$. 
In the \textsl{Interval Inform Step}, each node receives at most one pulse, which again yields a pulse complexity $O(N)$. 
\end{proof}

The next lemma follows directly from the $\intervalConstruction$ procedure: 
\begin{lemma} \label{lem:kt1_matched}
The $\intervalConstruction$ procedure adds $N/\log N$ edges to the matching.
Moreover, every node in $L_i$ and $R_i$ knows that it is a member of the respective interval, and knows which of its neighbors that are part of the same interval.
\end{lemma}

We are now ready to combine the previous lemmas to obtain the main result of this subsection:
}
\begin{theorem}
$\fastIntervalMatching$ computes a perfect matching in $O(\log n)$ rounds and with $O(n)$ pulse complexity, assuming the $\kt_1$ model.  
\end{theorem} 
\onlyLong{
\begin{proof}
Following Lemma~\ref{lem:kt1_messages} and Lemma~\ref{lem:kt1_matched}, we only need to show the round and pulse complexity of Parallel Interval Matching procedure. 
In this procedure, each node in $L_j$ receives one pulse from the leader in $R_j$ and sends one pulse to a node in $R_j$. 
Hence, the pulse complexity is $O(s)$ per interval and hence $O(N)$ in total. 
This procedure takes $O(s)$ rounds, hence, $O(\log n)$ round complexity.
\end{proof}
\subsection{Perfect Matching in $O(\log^{*}n\cdot\log\log n)$ Rounds} \label{sec:kt1_recursive}
}
We now employ procedure $\intervalConstruction$\onlyLong{ from Section~\ref{sec:kt1_simple}}\xspace as a building block for obtaining an exponential improvement in the time complexity, while increasing the message complexity by a factor of $\log^*n$.

\onlyLong{
Our algorithm proceeds recursively as follows. 
We first perform the $\intervalConstruction$ procedure on the original bipartite network to obtain the intervals $L_0,\ldots,L_{n/\Theta\lt( \log n \rt)}$ and $R_0,\ldots,R_{n/\Theta\lt( \log n \rt)}$, as described in Section~\ref{sec:kt1_simple}.
Then, we consider each pair $(L_i,R_i)$ as a new bipartite network $G_i^{(1)}$, and apply the $\intervalConstruction$ algorithm to $G_i^{(1)}$, in parallel for every $i \in [n/\Theta\lt( \log n \rt)]$. 
Note that Lemma~\ref{lem:kt1_matched} ensures that each node knows which interval it is part of.
We stop this recursion after a depth of $\log_2^{*}n$, which is the number of times the logarithm function must be iteratively applied before the result is less than or equal to $1$.
}
\begin{theorem} \label{lem:kt1_recursive}
There is a deterministic algorithm that computes a perfect matching in $O\lt( \log^{*}n\cdot\log\log n \rt)$ rounds and with a pulse complexity of $O\lt( n \log^{*}n \rt)$ under the $\kt_1$ assumption.
\end{theorem}
\onlyLong{
\begin{proof} 
As the algorithm proceeds recursively, we call the original complete bipartite graph of $2n$ nodes the \emph{depth-$0$ network}.
Then, after the first execution of $\intervalConstruction$, we partition the network into $n/\log n$ bipartite networks $G_i^{(1)}$, which we call the \emph{depth-$1$ networks}, and observe that each depth-$1$ network comprises $O(\log n)$ nodes. 
Then, we apply $\intervalConstruction$ on each of these depth-$1$ networks, (in parallel), and the original network is now partitioned into $n/\log\log n$ depth-$2$ networks, each of $O(\log \log n)$ nodes. 
In general, for $d \geq 1$, executing $\intervalConstruction$ for each of the depth-$(d-1)$ networks, results in the partition of the network into $n/\log^{(d)}n$ depth-$d$ networks, each consisting of $O(\log^{(d)}n)$ nodes. 

\begin{claim} \label{cl:depthd}
  The execution of $\intervalConstruction$ on the depth-$d$ networks, in parallel, takes $O(\log^{(d+2)} n)$ rounds and sends $O(n)$ pulses.
  Moreover, as a result, $n/\log ^{(d+1)}n$ edges are being matched.   
\end{claim}
\begin{proof} 
To prove the claim, note that each depth-$d$ networks has $2N = O(\log^{(d)}n)$ nodes. Lemma~\ref{lem:kt1_messages} shows that the execution of $\intervalConstruction$ on a depth-$d$ network  takes $O(\log \log N) = O(\log^{(d+2)} n)$ rounds and sends $O(N) = O(\log^{(d)}n)$ pulses. 
Since there are $n/\log^{(d)}n$ depth-$d$ networks, the total number of pulses across all the depth-$d$ networks is $O(n)$. 
Lemma~\ref{lem:kt1_matched} tells us that $N/\log N = \log^{(d)}n/\log^{(d+1)}n$ edges are matched in each depth-$d$ network. This applies to all the $n/\log^{(d+1)}n$ depth-$d$ networks, which results in a total of $n/\log^{(d)}n$ edges that are matched.
\end{proof}

Equipped with Claim~\ref{cl:depthd}, we are now ready to complete the proof of the theorem.
When $d = \log^{*}n$, it follows that $n/\log ^{(\log^* n)}n = n$ edges are matched, and we are done. 
Moreover, Claim~\ref{cl:depthd} tells us that the total number of rounds is at most 
\[
\sum_{d=0}^{\log^{*}n} O(\log^{(d+2)}n) = O\lt( \log^{*}n\cdot\log\log n \rt),
\]
whereas the number of pulses can be upper-bounded by $O\lt( n \log^{*}n \rt)$. 
\end{proof}

In the distributed load balancing problem, we are given a bipartite graph and the goal is to assign the nodes on the left (the ``balls'') to the nodes on the right (the ``bins'') in a way such that the maximum load of each bin is minimized.
Since a perfect matching yields the optimal load assignment, we obtain the following:

\begin{corollary} \label{cor:load_balancing}
Under the $\kt_1$ assumption, deterministic distributed load balancing is possible with a bin load of $1$ in $O\lt( \log^{*}n\cdot\log\log n \rt)$ rounds and  $O\lt( n \log^{*}n \rt)$ messages. 
\end{corollary}
}
\onlyShort{\section{A Lower Bound for Deterministic Algorithms }\vspace{-0.25cm}}
\onlyLong{ \section{A Lower Bound for Deterministic Algorithms in the Port Numbering Model}} \label{sec:lb}

In this section, we return to the port numbering model and show that randomization is crucial for any perfect matching algorithm that uses a small number of communication links. 
We make several assumptions that strengthen our lower bound, namely, we consider the standard $\local$ model~\cite{peleg}, where nodes may send messages of arbitrary size in each round, and 
deterministically label the nodes in $L$ with IDs from $\set{1,\dots,n}$ and the nodes in $R$ with IDs from $\set{n+1,\dots,2n}$.  

\begin{theorem} \label{thm:upper_complete}
   Any deterministic algorithm that constructs a perfect matching in the complete bipartite graph of $2n$ nodes has a message complexity of $\Omega(n^2)$ in the $\local$ model.
\end{theorem}
\onlyLong{ 

\begin{proof}
Note that each node has $n$ ports, over which it receives and sends messages, and, as explained in Section~\ref{sec:intro}, a node is unaware to which node a port is connected until a message is sent or received over that port. 
Moreover, to output the matching, we assume that, for every edge $\set{u,v}$ that is matched, node $u$ outputs the port connecting to $v$ and vice versa.

Note that we show our lower bound for general algorithms, where any node in either $L$ or $R$ may communicate with any number of nodes on the other side in a given round.
Consider some arbitrary enumeration $\mathcal{E}_L$ of the nodes in $L$ and also an enumeration $\mathcal{E}_R$ of the nodes in $R$. 
Below, when referring to the $i$-th node in $L$ (or $R$), we mean the $i$-th node with respects to these enumerations.

Let $ALG$ be any deterministic algorithm that outputs a perfect matching using at most $o(n^2)$ messages. 
Since any deterministic algorithms must correctly work on all port mapping between the nodes, we can specify the port mapping adaptively, i.e., depending on the messages sent by the algorithm. 
That is, our goal will be to describe a strategy for the adversary to obtain such a port mapping that results in $\Omega\lt( n^2 \rt)$ messages.
We say that a port $p$ of a node is \emph{unused} at the start of round $r$ if no message was sent or received over $p$ until the end of round $r-1$, and we say that $p$ is \emph{used} otherwise.

We first prove our lower bound assuming that the algorithm sends a message over every port that is output as part of the matching.
Moreover, we also assume that the nodes in $L$ send messages only in odd rounds and the nodes in $R$ send messages only in even rounds.
At the last paragraph, we will show how to remove these two restrictions.

Since we assume that nodes in $L$ and $R$ do not speak at the same time, it is straightforward to verify that the following port assignment rules for the adversary always yield a well-defined port assignment:
\begin{enumerate} 
    \item \textsl{Rule (L):} 
    Suppose a node $a_i \in L$ sends a message over an unused port $p$ in some odd round.
    We connect $p$ to node $b_j \in R$, where $j$ is the smallest index with respect to $\mathcal{E}_R$, such that no message was sent between $a_i$ and $b_j$.
    If $a_i$ sends messages over $k$ unused ports in the same round, we apply this rule $k$ times, i.e., these messages will be directed to the first (w.r.t.\ $\mathcal{E}_R$) nodes that have not yet communicated with $a_i$
    \item \textsl{Rule (R):}
    If a node $b_i \in R$ sends a message over an unused port in some even round, we proceed analogously to Rule~(L), with the roles of $L$ and $R$ being swapped, i.e., these messages are directed to the first nodes in the order defined by $\mathcal{E}_L$ that $b_i$ has not yet communicated with.
\end{enumerate}

Consider an execution where the ports are connected according to Rules (L) and (R).
Let $M$ be the matching output by the algorithm.
Since every used port corresponds to at least one sent message, and recalling that the algorithm sends at most $o\lt( n^2 \rt)$ messages, it follows that there exists a subset $M' \subseteq M$ of size $m' := |M'| \ge \frac{3}{4}n$, such that every node that occurs in some pair in $M'$ has at most $\frac{n}{32}$ used ports at the end of the execution.
Without loss of generality, let $M'=\set{(a_{i_1},b_{j_1}),\dots,(a_{i_{m'}},b_{j_{m'}})}$, i.e., for every $k \in \{1, \ldots, m'\}$, nodes $a_{i_k}$ and $b_{j_k}$ are matched to each other, where node $a_{i_k}$ refers to the $i_k$-th node in the enumeration $\mathcal{E}_L$, and, similarly, $b_{j_k}$ refers to the $j_k$-th node in $\mathcal{E}_R$.
Let $L' = \{a_{i_1},\dots, a_{i_{m'}}\}$, and $R' = \{b_{j_1},\dots, b_{j_{m'}}\}$. 
We define $L_{\ge n/8}' = \set{a_{i_k} \mid k \in [\lt\lfloor n/8 \rt\rfloor,m]}$, i.e., $L_{\ge n/8}'$ contains the $(m' - \lt\lfloor n/8 \rt\rfloor)$-highest ranked nodes of $L'$ that are in $M'$, where the ranking is with respect to the ordering $\mathcal{E}_L$. 
Note that $|L_{\ge n/8}'| \ge \frac{n}{2}$.
We define $R_{\ge n/8}'$ analogously, and also define $R_{< n/8}' = R' \setminus R_{\ge n/8}'$.

\begin{claim} \label{cl:hall}
The nodes in $L_{\ge n/8}'$ do not exchange any messages with the nodes in $R_{\ge n/8}'$. 
\end{claim}
\begin{proof} 
  By assumption the nodes in $L'$ send or receive at most $\frac{n}{32}$ messages each.
  Thus, we can pessimistically assume that, at any point in the execution, each node $a \in L'$ has received at most $\frac{n}{32}$ messages.
  Let $R_a$ denote the set of nodes that sent messages to $a$. 
  By Rule~(L), the (at most) $\frac{n}{32}$ messages that $a$ sends over unused ports throughout the execution will be addressed to the first $\frac{n}{32}$ nodes (w.r.t.\ $\mathcal{E}_R$) that $a$ has not yet communicated with.  
  Thus, it follows that only the first $|R_a|+\frac{n}{32} \le \frac{n}{16}$ nodes in $R$ can receive a message from $a$, and consequently, none of the
  nodes in $R_{\ge n/8}'$ receive a message from $a$.
\end{proof}

Now consider the induced bipartite subgraph $H'$, which has the same vertex set as the nodes occurring in $M'$ and where every edge corresponds to a message sent by the execution of algorithm $ALG$ between the nodes in $L'$ and $R'$, ignoring edge directions. 
(We omit messages that are sent to nodes not in $L'$ or $R'$.)
Since we assume that $ALG$ sends a message over every matched edge, $M'$ is a perfect matching in $H'$.
From Claim~\ref{cl:hall} we know that the nodes in $L'_{\ge n/8}$ only communicate with nodes in $R'_{< n/8}$.
However, the neighborhood of $L'_{\ge n/8}$ is of size at most $n/8$, whereas we have $|L'_{\ge n/8}| \ge |L'| - \frac{n}{8} \ge \frac{5n}{8}$.
Thus, we have arrived at a contradiction to Hall's Marriage Theorem~\cite{Hall1935}, which states that a perfect matching only exists if the neighborhood of every set of nodes is at least as large as the set itself.

Finally, we remove the restrictions that we assumed on algorithm $ALG$. 
First, observe that any algorithm that does \emph{not} satisfy the requirement that the nodes in $L$ only speak in odd rounds and the nodes in $R$ only speak in even rounds can be transformed into one that does satisfy the requirement while increasing the time complexity by a factor of two and without any overhead in terms of message complexity. 
Next, suppose that there is an algorithm $ALG'$ that correctly computes a perfect matching by outputting the port numbers of the matched edges, but does not send a message over every matched edge, while sending only $o(n^2)$ messages. 
We simply extend $ALG'$ by instructing each node to send a message over its matched port number. Thus, we obtain an algorithm with a message complexity of $o\lt( n^2 \rt) + 2n = o\lt( n^2 \rt)$ messages that sends a message over all matched edges.
This completes the proof of Theorem~\ref{thm:upper_complete}.
\end{proof}

}
\onlyLong{
\section{Conclusion}

We have initiated the study of the perfect matching problem in the link activation model in the synchronous setting.
In future work, it will be interesting to extend these results to other graph problems such as maximal independent sets. 
Moreover, an open question is whether a perfect matching is possible by sending only $O(n)$ pulses in the port numbering model.
}
{
\paragraph{Acknowledgements.}
Hugo Mirault and Peter Robinson were supported in part by National Science Foundation (NSF) grant CCF-2402836.
Ming Ming Tan was supported in part by National Science Foundation (NSF)
grant CCF-2348346.
Xianbin Zhu was partially supported by a grant from the Research Grants Council of the Hong Kong Special Administrative Region, China [Project No.\ CityU 11213620].}
%onlyLong{
\bibliographystyle{alpha}
\bibliography{references,refs}
%\onlyShort{
%\input{main.bbl}
%}
\end{document}